\documentclass[11pt,a4paper]{amsart}

\usepackage{graphicx}
\newtheorem{theorem}{Theorem}
\newtheorem{lemma}[theorem]{Lemma}
\newtheorem*{lemma*}{Lemma}
\newtheorem{proposition}[theorem]{Proposition}

\newtheorem{definition}[theorem]{Definition}

\usepackage[utf8]{inputenc}
\usepackage[english]{babel}
\usepackage{latexsym}
\usepackage{amssymb}
\usepackage{amsmath}

\newcommand{\Z}{\mathbb{Z}}
\newcommand{\R}{\mathbb{R}}

\newcommand{\E}{\mathbb{E}}

\newcommand{\caO}{{\mathcal O}}

\newcommand{\hf}{\frac{_1}{^2}}

  \newtheorem{remark}[theorem]{Remark}

\newcommand{\beq}{ \begin{equation} }
\newcommand{\eeq}{ \end{equation} }
\newcommand{\bet}{ \begin{theorem} }
\newcommand{\eet}{ \end{theorem} }

\newcommand{\baq}{\begin{eqnarray}}
\newcommand{\eaq}{\end{eqnarray}}

\title[Analysis of the LME recursion]{Towards Rigorous Analysis of the Levitov-Mirlin-Evers recursion}

\author[Y.V. Fyodorov]{Y. V. Fyodorov}
\address{Queen Mary University of London, School  of Mathematical Sciences, London, E1 4NS, United Kingdom}
\email{y.fyodorov@qmul.ac.uk}
\author[A. Kupiainen]{A. Kupiainen }
\address{University of Helsinki, Department of Mathematics and Statistics, P.O. Box 68 , FIN-
00014 University of Helsinki, Finland}
\email{antti.kupiainen@helsinki.fi}
\author[C. Webb]{C. Webb}
\address{Department of mathematics and systems analysis, Aalto University, PO Box 11000, 00076 Aalto, Finland}
\email{christian.webb@aalto.fi}
\date{\today}

\begin{document}
\begin{abstract}
This paper aims to develop a rigorous asymptotic analysis of an approximate renormalization group recursion for inverse participation ratios $P_q$ of critical powerlaw random band matrices. The recursion goes back to the work by Mirlin and Evers \cite{em} and earlier works by Levitov \cite{lev1,lev2} and is aimed to describe the ensuing multifractality of the eigenvectors of such matrices. We point out both similarities and dissimilarities between the LME recursion and those appearing in the theory of multiplicative cascades and branching random walks and show that the methods developed in those fields can be adapted to the present case. In particular the LME recursion is shown to exhibit a phase transition, which we expect is a freezing transition, where the role of temperature is played by the exponent $q$. However, the LME recursion has features that make its rigorous analysis considerably harder and we point out several open problems for further study.
\end{abstract}

\maketitle

\section{Introduction}

The goal of this paper is to study the asymptotic behavior of a stochastic branching recursion (which we refer to as the Levitov-Mirlin-Evers (LME) recursion) conjectured to be related to multifractal properties of the Anderson transition. More precisely,
building on an earlier work of Levitov  \cite{lev1,lev2} Mirlin and Evers  argued  in \cite{em}
that in the so-called critical power law random banded matrix model (PRBM), for small values of an auxiliary variance parameter, the inverse participation ratio (essentially the $\ell^p$ norm) of a typical eigenvector of the random matrix approximately satisfies the LME-recursion. By analyzing this recursion  they further argued that the inverse participation ratio scales in an anomalous (multifractal) way with the size of the matrix. Let us also mention that a related, though less explicitly developed renormalization procedure was outlined for essentially the same model in \cite{ps}, with similar conclusions.

\vspace{0.3cm}

Our goal is to rigorously study this recursion and prove the conjectured scaling behavior. Our approach is based on the fact that this recursion resembles ones appearing in the study of weighted branching processes such as the branching random walk. These recursions, their applications and related models appear for example in other areas of probability, the study of algorithms, mathematical physics and finance and have been studied extensively \cite{mandelbrot,kah,big,rderev,gmc,bkm}. While being similar to recursions arising in other applications, the LME-recursion has also important differences. In particular, we are not able to prove in full generality the type of results that are available for more traditional branching recursion and this leaves various open questions that might be interesting for understanding better the multifractal geometry of these eigenvectors.

\vspace{0.3cm}

Thus in addition to studying the model itself, this note aims to bridge the gap between the Theoretical
 Physics and Probability communities in both directions. Namely, on one hand we are presenting to readers from the physics community with some of the results and tools used in the Probability communities to study some multifractal geometric objects. On the other hand we present to readers from the Mathematics community with a new type of stochastic recursion which has important applications in random matrix theory and its relations to physics of disordered systems and is an interesting and rich object in itself. In particular, there are many open questions we are not yet able to answer.

\vspace{0.3cm}

In Section 2 we briefly recall the critical power law random band matrix model and the heuristic derivation of the LME recursion for the  inverse participation ratios. In Section 3 we recall branching recursions arising in the study of branching random walks and the type of results and tools one has there. Section  4 studies the scaling properties of the moments of the solution of the LME recursion. We deduce the existence of a phase transition in the model, namely for small values of the index of the participation ratio, the participation ratio lives on the scale of its mean, while for large values it does not.
Using the knowledge on the scaling behavior of the moments, we are able to prove (using the method of moments) that for small values of the participation ratio index, the correctly normalized inverse participation ratio converges in law to a non-trivial random variable whose law can be characterized in terms of a non-linear integro-differential equation. Finally,
in Section 8 we discuss some open questions suggested by the connection with branching random walks.

\vspace{0.3cm}

{\bf Acknowledgements:} Y.F and A.K. are grateful to the Institute for Advanced Study, Princeton
for its kind hospitality during the initial stage of the project.
 \rm The research by A.K. and C.W.  was supported by the Academy of Finland, and Y.F. was supported by
  EPSRC grant EP/J002763/1 ``\textsf{Insights into Disordered Landscapes via Random Matrix Theory and Statistical Mechanics}''. C.W.  wishes to thank Nathana\"el Berestycki for interesting discussions related to this note. We also wish to thank an anonymous referee for careful reading of a preliminary version of the article, as well as helpful comments.

\section{The model}\label{sec:model}

To motivate our stochastic recursion, we recall the model and argument presented in \cite{em}, based on the earlier work in \cite{lev1}. We consider a power-law random band matrix model \cite{prbm} - that is random $N\times N$ real symmetric matrices $H$ with independent entries whose variance  decays with the distance from the diagonal in a power law fashion:
\begin{equation}\label{model}
\E(H_{ij})^2=\left\{
 \begin{array}{cl}
  1 & |i-j|<b\\
   (\frac{b}{|i-j|} )^{2\alpha}& |i-j|>b
 \end{array} \right. ,%.\frac{1}{1+\left(\frac{|i-j|}{b}\right)^{2\alpha}}.
\end{equation}
with $b$ being a coupling constant controlling for a fixed $\alpha$ the relative strength of fluctuations of the entries close to the main diagonal and away from it. In particular, for $b\gg 1$ there is a well-defined band of the width $b$ around the diagonal where the variance remains constant, and then decays in a powerlaw way outside that band. In the opposite case $b\to 0$ the condition $|i-j|<b$ implies $i=j$
so that the matrix has a dominating diagonal with a small admixture of off-diagonal entries with powerlaw-decaying amplitude.

Such a model for $\alpha=1$ is considered in Theoretical Physics literature as a paradigmatic example of {\it critical} random matrix models with multifractal eigenvectors, see e.g. \cite{Kravtsov} for a recent discussion. As such it played in the last decade an important role in attempts to understand generic properties of wave functions for systems at the point of the Anderson localization transition where such multifractality is conjectured to be one of the most characteristic features  \cite{EMrev}.
Let us note that a few other critical random matrix models with multifractal eigenvectors attracted much of attention recently \cite{FOR2009,BG2011,MG2011}.

It was first suggested in \cite{prbm} that in the limit of large control parameter $b\gg 1$ and Gaussian randomness  the random matrix model (\ref{model}) undergoes an Anderson-type transition when changing the exponent parameter $\alpha$. Namely, via a mapping to a non-linear $\sigma$-model it was argued that whereas for any $\alpha>1$ the eigenvectors of the underlying matrices are  localized, they are delocalized for $\alpha<1$, and precisely at the critical value $\alpha=1$ eigenvectors show nontrivial multifractal behavior. To quantify this statement consider a typical eigenvector
$\psi=(\psi_1,...,\psi_N)$ and given a parameter $q>0$ define the inverse participation ratio in terms of the $\ell^{2q}$ norm as

\begin{equation}\label{ipr}
P(q,\psi)=\|\psi\|_{2q}^{2q}=\sum_{i=1}^N|\psi_i|^{2q}.
\end{equation}
where we normalize the $\ell^2$ norm to $1$ i.e. $P(1,\psi)=1$. One expects $P(q)$ to
scale with volume as
\begin{equation}\label{eq:scaling}
P(q,\psi)\propto N^{-d(q)(q-1)}
\end{equation}
with $d(q)=0$  for localized states,  $d(q)=1$ for extended states and
$0<d(q)<1 $ for multifractal states.

We consider in this paper only the critical case $\alpha=1$ in \eqref{model} leading to multifractal eigenfunctions.
Let us also for simplicity take periodic boundary conditions i.e. we let the indices $i\in\Z_N$, the integers modulo $N$, and let $|i-j|$  in  \eqref{model} stand for  the distance on $\Z_N$ and similarly the addition $i+j$.
The LME renormalization group deals with the same model (\ref{model}) in the case of a small control parameter
 $b\to 0$ which is exactly opposite to that considered in \cite{prbm}. Nevertheless, by heuristically deriving and analyzing an approximate renormalization group flow
Mirlin and Evers concluded that  for the critical value $\alpha=1$ a nontrivial multifractality of eigenvectors survives for any small $b>0$, although all $d(q)$ in this limit remain, non-surprisingly, parametrically close to the localized limit  $d(q)=0$.

 Essentially, the LME procedure aims to construct the orthogonal matrix $O$ diagonalizing $H$,
$O^THO={\rm diag}(E_i)$ inductively by finding matrices  $O_n$ s.t.
$$(O_n^THO_n)_{ij}=0,\ \ 1\leq |i-j|\leq n.
$$
 % Let $H^n_{ij}:=H_{ij}1_{|i-j|\leq n}$
%be the matrix obtained from  $H$  when its  matrix elements with distance $>n$ from diagonal are set to zero. Let $O_n$ diagonalize $H^n$. We construct $O_n$ inductively in $n$.

Obviously $O_0=1$ so let $n=1$. We have that $E^0_i=H_{ii}$ are i.i.d
with variance 1. Hence typical eigenvalue differences are $E^0_i-E^0_j=\mathcal{O}(1)$. Moreover, the variance of the perturbation $h_{ij}:=H_{ij}1_{|i-j|= 1}$ is $b^2$ so that typically $h_{ij}=\mathcal{O}(b)$. If these conditions were satisfied the matrix  $O_1$ would be  $O_1=1+m$ where $m=\caO(b)$. Indeed, to  first order in the perturbation $h$
\beq
m_{ij}=\frac{h_{ij}}{E^0_j-E^0_i}+\dots\label{onij}.
\eeq
However, when
$E^0_i-E^0_j=\mathcal{O}(b)$ perturbation theory breaks down, and such an event is called a {\it resonance}.  Pick $a>0$ and define the {\it resonance set at scale 1} by
$$
R_1=\{i: \exists j, |i-j|=1,|E^0_i-E^0_j|\leq b^a\ \}.
$$
With high probability, as $b\to 0$ the set $R_1$ consists of isolated nearest neighbor pairs separated by distance $\caO(b^{-a})$. Under such an event the matrix $h|_{R_1}$ is diagonalized by the matrix
\beq\label{odec0}
O_{R_1}:=\oplus_{\{i,i+1\}\subset R_1}O^{i}_1
\eeq
where $O^{i}_1$ diagonalizes the resonant matrix
\begin{equation}\label{reso}
\left(\begin{array}{cc}
E^0_i & H_{ii+1} \\
H_{ii+1}  & E^0_{i+1}
\end{array}\right).
\end{equation}
The LME RG will now be defined by making  the approximation
\beq\label{odec}
O_1=O_{R_1}\oplus 1_{R_1^c}
\eeq
and  approximating $O_1^THO_1$ by ${\rm diag}(E^1_i)+H^{1}$ where $E^1_i=E^0_i$ when $i\in R_0^c$ and $E^1_i,E^1_{i+1}$ are the eigenvalues of \eqref{reso} if $i,i+1\in R_1$ and $H^{1}_{ij}=H_{ij}1_{|i-j|>1}$.

Note that the true $O_1$ has of course  a more complicated structure. In the non resonant region its off-diagonal matrix elements $O_{1,ij}$ will not be exactly zero, but small and exponentially decaying in the separation $|i-j|$. Moreover the resonant and non resonant regions are not completely decoupled. However for small $b$ \eqref{odec} should capture the essence of $O_1$. Also, $O_1^THO_1$ will include matrix elements with
$|i-j|=1$ and further renormalize the diagonal elements $E^1_i$. These corrections are subleading in $b$ and are dropped in the LME approximation.

Inductively, we approximate  $O_{n}^THO_{n}$ by ${\rm diag}(E^{n}_i)+ H^n$ where $H^{n}_{ij}=H_{ij}1_{|i-j|>n}$. We define  the {\it resonance set at scale n+1} by
$$
R_{n+1}=\{i: \exists j, |i-j|=n+1,|E^{n}_i-E^{n}_j|\leq (\frac{_b}{^{n+1}})^a\ \}.
$$
and set
$$
O_{n+1}=O_{n}o_n
$$
where $o_n$ is defined as $O_1$ above by diagonalizing on the resonant set $
R_{n+1}$ the matrix
\begin{equation}\label{reson}
M_{ij}^n:=
\left(\begin{array}{cc}
E^n_i & h_{ij} \\
h_{ij}  & E^n_{j}
\end{array}\right)
\end{equation}
where $h_{ij}=H^n_{ij}1_{|i-j|=n+1}$.
$h_{ij}$ have variance $b^2/(n+1)^2$. This time the resonant set has  density $\caO((\frac{b}{n+1})^{a})$.

If we set  $H^n=0$ the eigenvectors of $H$ are given by $\psi^n_i=O_n\psi^0_i$ where $\psi^0_i$ are the canonical  basis vectors of $\R^N$. In case of a resonant pair $i,j$ write the restriction of $o_n$ to $V_{ij}:={\rm span}\{\psi^0_i, \psi^0_j\}$ as
 \begin{equation}\label{evrec}
 o_n|_{V_{ij}}= \left(\begin{array}{cc}
\cos\theta_n & \sin\theta_n \\
 -\sin\theta_n & \cos\theta_n
\end{array}\right)
\end{equation}
so that 
\begin{align*}
\psi^{n+1}_i&=\cos\theta_n\psi^{n}_i + \sin\theta_n\psi^{n}_j\\
\psi^{n+1}_j&=\cos\theta_n\psi^{n}_j - \sin\theta_n\psi^{n}_i.
\end{align*}

We point out that the random variable $\theta_n$ here depends on the resonant pair $(i,j)$, but for all pairs, the distribution is the same.

%In case of a resonant pair $i,j$, we get
%\begin{equation}\label{evrec}
%\left(\begin{array}{rl}
%\psi^{n+1}_i  \\
%\psi^{n+1}_j
% \end{array} \right)=o_n\left(\begin{array}{rl}
%\psi^{n}_i  \\
%\psi^{n}_j
 %\end{array} \right)
% \end{equation}
 Hence % $(\psi^{n+1}_i, \psi^{n+1}_j)^T= o_n(\psi^{n}_i,\psi^{n}_j)^T$
%where $o_n$ diagonalizes the matrix
we get for the the inverse participation ratios
\beq\nonumber
P(q;\psi_i^{n+1})=\|\cos\theta_n\psi^{n}_i+\sin\theta_n\psi^{n}_j\|_{2q}^{2q}
\eeq
Note that $\psi^{n}_i$ has support in the original basis in a ball of radius $k$ around $i$ where $k$ is the largest scale $k\leq n$ s.t. $\psi^{k}_i$ was resonant. Hence with high probability $\psi^{n}_i$ and $\psi^{n}_j$ have disjoint supports and  on that event
\beq\label{iter}
P(q;\psi_i^{n+1})=(\cos^{2}\theta_n)^qP(q;\psi_i^{n})+(\sin^{2}\theta_n)^qP(q;\psi_j^{n}).
\eeq

\begin{remark}
The terms $(\sin^{2}\theta_n)^q$ and $(\cos^2\theta_n)^q$ will appear repeatedly and for notational brevity, we will write them as $\sin^{2q}\theta_n$ and $\cos^{2q}\theta_n$ (the possible ambiguity in the latter notation being when $\sin\theta_n$ or $\cos\theta_n$ are negative).
\end{remark}

The eigenvalues $E^n_i$ and $E^n_j$ are functions of the original $E^0_k$, $H_{lm}$ with
disjoint support in the indices and hence independent in our approximation. Obviously the law of
$P(q;\psi_i^{n})$ is independent of $i$. We are therefore led to consider the iteration for random variables $P_n(q)$:
\begin{equation}\label{lemrec}
P_{n+1}(q)\stackrel{d}{=}\sin^{2q}\theta_n P_n^{(1)}(q)+\cos^{2q}\theta_n P_n^{(2)}(q),
\end{equation}
 where $P_{n}^{(i)}$ are mutually independent copies of $P_n(q)$, independent  of  $\theta_n$  and $P_1(q)=1$. Note that \eqref{lemrec} differs from \eqref{iter} in that the latter was argued to hold (approximately) provided the resonant set consists of disjoint pairs whereas in the former we extend that relation everywhere. In fact shortly we will argue that this change does not affect the multifractal behavior.

 To understand the recursion \eqref{lemrec}, let us derive the law of $\theta_n$.  First note that $\theta_n$, $\theta_n+\hf\pi$ and $\hf\pi-\theta_n$ lead to the same eigenspaces. Hence we may restrict $\theta_n$ to the interval  $\theta_n\in[-\frac{\pi}{4}, \frac{\pi}{4}]$. Introducing the variable
$$
\tau_n={\delta E^n}/{h_{ij}}$$
where $\delta E^n=\hf(E^n_i-E^n_j)$ the eigenvectors %$v_\pm=(a_\pm,b_\pm)^T$
 $v=(a,b)^T$ of \eqref{reson} satisfy
$$
(\tau_n\pm\mathrm{sgn}(h_{ij})\sqrt{1+\tau_n^2})a+b=0.
$$
Writing
\begin{equation}\label{tau}
\tau_n=-\cot 2\theta
\end{equation}
some algebra leads to $\frac{a}{b}=-\cot\theta$ or  $\frac{a}{b}=\tan\theta$. Hence we may identify $\theta$ with $\theta_n$.

Let the law of $t:={\delta E^n}$  have density $\rho(t)$ and let $h_{ij}$ have
the law of $\epsilon v$ where $\epsilon=\frac{b}{n}$ and  $v$ has density $\gamma(v)$. We suppose
$\rho$ and $\gamma$  have all positive moments and that $\|\rho'\|_\infty<\infty$. Then  for a bounded continuous function $f$
\beq\label{tanlaw}
\E(f(\theta_n))=\int_\R dt \rho(t)\int_\R dv \gamma(v)f\left(-\hf\cot^{-1}(\frac{_t}{^{\epsilon v}})\right)
\eeq
Changing  variables by % $\frac{t}{{\epsilon v}}+(1+(\frac{t}{{\epsilon v}})^2)^\hf=\tan\theta$, $\theta\in[0,\frac{\pi}{2}]$ i.e. setting
$v= -\frac{t}{{\epsilon }}\tan 2\theta$  we get
\baq\label{tanlaw1}
\E(f( \theta_n))&=&\frac{2}{\epsilon}\int_{-\frac{\pi}{4}}^{\frac{\pi}{4}}\frac{d\theta}{\cos^22\theta}
\int_\R dt \rho(t)|t| \gamma(-\frac{_{\tan 2\theta} }{^\epsilon}t)f(\theta)\nonumber\\
&=&\epsilon\int_{-\frac{\pi}{4}}^{\frac{\pi}{4}}\frac{d\theta}{\sin^22\theta}
\chi_\epsilon(\theta)f(\theta)
\eaq
where
\beq\label{tanlaw2}
 \chi_\epsilon(\theta)=2\int_\R dt \rho(-\frac{_\epsilon}{^{\tan 2\theta }}t)|t |\gamma(t).
\eeq
Note that $ \chi_\epsilon$ regularizes the integrand  in \eqref{tanlaw1} around the singularity at $\theta=0$. Indeed, %for $\theta\in  [0,\epsilon]\cup [\frac{\pi}{2}-\epsilon, \frac{\pi}{2}]$ we get
by change of variables
\beq\label{chi1}
 \chi_\epsilon(\theta)\leq 2(\frac{_{\tan 2\theta} }{^\epsilon})^2 \|\gamma\|_\infty\int_\R dt \rho(t)|t |\leq C(\frac{_{\theta} }{^\epsilon})^2
\eeq
 Moreover, we have

\begin{lemma}\label{sing}Let $f$ be bounded on   $[-\frac{\pi}{4}, \frac{\pi}{4}]$ with $f(\theta)=\caO(|\theta|^\alpha)$ as $\theta\to 0$, with $\alpha>1$. Then
$$
\E(f( \theta_n))=\beta\frac{_b}{^n}\int_{-\frac{\pi}{4}}^{\frac{\pi}{4}}
\frac{d\theta}{\sin^22\theta}
f(\theta)+\caO((\frac{_b}{^n})^{\alpha})
$$
with $\beta=2\rho(0)\int_\R dt |t |\gamma(t)$.
\end{lemma}
\begin{proof} We have by \eqref{chi1}
$$
\epsilon\int\frac{d\theta}{\sin^22\theta}
\chi_\epsilon(\theta)f(\theta)1_{|\tan 2\theta|\leq\epsilon}\leq C\epsilon^{-1}\int_{-\epsilon}^\epsilon|\theta|^\alpha\leq C\epsilon^\alpha.
$$
Moreover, by our assumption on $f$,
$$
\epsilon \int\frac{d\theta}{\sin^2 2\theta}f(\theta)1_{|\tan 2\theta|\leq \epsilon}\leq C\epsilon^\alpha.
$$
From \eqref{tanlaw2} we have
\beq\label{chi2}
| \chi_\epsilon(\theta)-\beta|\leq C\frac{_\epsilon}{^{|\tan 2\theta|} }\|\rho'\|_\infty\leq C\frac{_\epsilon}{^{|\theta|} }
\eeq
so that
$$
\epsilon\int\frac{d\theta}{\sin^22\theta}
|\chi_\epsilon(\theta)-\beta|f(\theta)1_{|\tan 2\theta|\geq\epsilon}\leq C\epsilon^\alpha.
$$
Putting things together gives the claim.

\end{proof}
\begin{remark}
The parameter $\beta$ depends on $\rho(0)$ i.e. the density of the eigenvalue differences $\delta E_n$ at $\delta E_n=0$ and on $\E|v|$. Note that we have (in our approximation)
$$
|\delta E_{n+1}|=\sqrt{\delta E_n^2+h_n^2}
$$
where $h_n=h_{ij}$ with $|i-j|=n$. Therefore $\beta_n$ tends to a limit $\beta$ as $n\to\infty$.
The multifractal exponents turn out to be proportional to $\beta$. Since it plays no further role in the analysis we will set it to one in what follows. This will change some of our results from those in \cite{em} by a factor of $8/\pi$.
\end{remark}
Motivated by these considerations we can now define the recursion which we will study in this paper.

\begin{definition}
We say that the law of $P_n(q)$ satisfies the LME recursion if
\begin{equation}\label{LME}
P_{n+1}(q)\stackrel{d}{=}\sin^{2q}\theta_n P_n^{(1)}(q)+\cos^{2q}\theta_n P_n^{(2)}(q),
\end{equation}
where $P_n^{(1)}$ and $P_n^{(2)}$ are independent copies of $P_n(q)$, $P_1(q)=1$ and $\theta_n$ is independent of $P_n^{(i)}(q)$ and distributed on $[-\frac{\pi}{4}, \frac{\pi}{4}]$ with density
\begin{equation}
r_n(\theta)=\frac{_b}{^n}\chi_{\frac{_b}{^n}}(\theta){(\sin2\theta)}^{-2}
\end{equation}
and $\chi$ is any non-negative function s.t. $\int r_n=1$ and  \eqref{chi1} and \eqref{chi2}  hold. We also recall that by $\sin^{2q}\theta_n$ and $\cos^{2q}\theta_n$ we mean $(\sin^2\theta_n)^q$ and $(\cos^2 \theta_n)^q$.

\end{definition}

\begin{remark}
Note that as $n\to\infty$, the law of $\theta_n$ degenerates to $\theta=0$ a.s. %the law of $\theta_n$ approaches that of a random variable that takes the value $0$ with probability $\frac{1}{2}$ and $\frac{\pi}{2}$ with probability $\frac{1}{2}$. Thus in the $n\to\infty$ limit, the variables $(\sin^{2q}\theta_n,\cos^{2q}\theta_n)$ converge in law to $(X,1-X)$, where $X$ is a Bernoulli-$\frac{1}{2}$ random variable and the recursion becomes degenerate and gives no information.
\end{remark}

\begin{remark}If we write $f_n$ for the density of the law of $P_n$, then the LME recursion in terms of $f_n$ is
\begin{align*}
f_{n+1}(p)=\int_{-\frac{\pi}{4}}^{\frac{\pi}{4}}
d\theta r_n(\theta)\int_0^\infty\int_0^\infty dp_1dp_2f_n(p_1)f_n(p_2)\delta(p-p_1\sin^{2q}\theta-p_2\cos^{2q}\theta),
\end{align*}
which is essentially a discrete version of equation (43) in \cite{em}.
Similarly we can express the recursion in terms of the Laplace transform of the law of $P_n$. Let us write $\psi_n(t)=\E(e^{-tP_n})$. Plugging the recursion into this gives
\begin{equation}
\psi_{n+1}(t)=\int_{-\frac{\pi}{4}}^{\frac{\pi}{4}}
r_n(\theta)\psi_n(t\sin^{2q}\theta)\psi_n(t\cos^{2q}\theta)d\theta.
\end{equation}
\end{remark}

\begin{remark}
Recall that we defined the resonant pairs by $ |E^n_i-E^n_j|\leq (b/n)^a$. Since $h_{ij}=\caO(1/n)$ this
implies that for non-resonant pairs $|\theta_n|=\caO(n^{a-1}) $ and therefore provided $2q>1$ which we will assume throughout the paper we may arrange $2q(1-a)>1$ and $ 2(1-a)>1$. This means that the contribution of the non-resonant pairs in the iteration \eqref{lemrec} is $\caO(n^{-1-a'}P_n)$ with $a'>0$ which will not contribute to the multifractal exponent.
\end{remark}

\begin{remark}\label{imbrie}
 Our approximations in deriving the LME RG were dropping corrections to the renormalized eigenvalues $E^{n+1}_i$ and the matrix $o_n$ in the non resonant region and ignoring multiple resonances. The perturbation theory in the non resonant region  produces off-diagonal terms to $o_n$ that are exponentially decaying (at the scale $n$) and down by inverse powers in $n$. Similarly, $E^{n+1}_i$ receive corrections that are non-local in the $E^n_k$ and $h$. Thus, even if no resonances were present the eigenvalues   $E^n_i$ are not independent but weakly correlated.  In the presence of resonances perturbation theory in the non resonant region does not decouple from the resonant one and there are further corrections that should be small due to the small probability of resonances. However to make a rigorous proof is very challenging. It would be very interesting to try to adopt the iterative scheme of \cite{Imbrie} to this problem.
\end{remark}
\section{Some background in recursions of branching random walks}

Let us consider an extremely simple example of a branching recursion related to a branching random walk. Let $V_1$ and $V_2$ be i.i.d. standard normal random variables  (compared to \eqref{lemrec}, these play the role of $\log \sin^2\theta_n$ and $\log\cos^2\theta_n$), $Z_0(\beta)=1$ ($Z$ corresponding to $P$), $\beta\in \R$ (corresponding to $q$) and define

\begin{equation}\label{cascade}
Z_{n+1}(\beta)\stackrel{d}{=}e^{\beta V_1}Z_n^{(1)}(\beta)+e^{\beta V_2}Z_n^{(2)}(\beta),
\end{equation}

\noindent where $Z_n^{(i)}$ are independent copies of $Z_n$ and they are independent from the $V_i$ as well. Compared to \eqref{lemrec}, we point out immediately some differences. In \eqref{lemrec}, the terms corresponding to $(e^{\beta V_1},e^{\beta V_2})$ are not independent, moreover their law depends on $n$ and as $n\to \infty$ and their joint law converges to a degenrate distribution (one of them is almost surely one and the other almost surely zero).

\vspace{0.3cm}

The asymptotic behavior of $Z_n$ and related objects has been studied under various generalizations of this recursion, see e.g. \cite{big, kah,kp,bram, bk, aid, as, ds, mad,cw,rl,dl,rderev}. Gaussianity is not important, one might have a random amount of terms in the recursion and the terms might even be correlated. Though much of the behavior in such recursions is universal \eqref{cascade} is perhaps the simplest recursion to study these universal features in. Moreover, even situations where the law of the terms corresponding to $V_i$ depends on $n$ has been studied in some relatively simple cases, but it seems that the situation where the limiting law of these has such a degenerate structure has not been studied rigorously.

\vspace{0.3cm}

Before briefly discussing the type of methods used to study recursions of this type, let us state the theorem describing the asymptotic behavior of $Z_n(\beta)$.

\begin{theorem}[\cite{kah,kp,big,bram,as,mad,cw,brv}]\label{th:cascconv}
For $\beta<\beta_c=\sqrt{2\log 2}$, $M_n(\beta)=\frac{Z_n(\beta)}{\E(Z_n(\beta))}=2^{-n} e^{-n\frac{\beta^2}{2}}Z_{n}(\beta)$ converges in law to an almost surely positive random variable $M(\beta)$, which satisfies the distributional equation

\begin{equation}
M(\beta)\stackrel{d}{=}\frac{1}{2}e^{\beta V_1-\frac{\beta^2}{2}}M^{(1)}(\beta)+\frac{1}{2} e^{\beta V_2-\frac{\beta^2}{2}}M^{(2)}(\beta),
\end{equation}

\noindent where $M^{(i)}(\beta)$ are independent copies of $M(\beta)$ and are independent of $V_j$.

\vspace{0.3cm}

For $\beta\geq \beta_c$, $M_n(\beta)$ converges in law to zero, but $\sqrt{n}M_n(\beta_c)$ converges in law to an almost surely positive random variable $M'(\beta_c)$ which satisfies

\begin{equation}
M'(\beta)\stackrel{d}{=}\frac{1}{2}e^{\beta_c V_1-\frac{\beta_c^2}{2}}(M'(\beta_c))^{(1)}+\frac{1}{2} e^{\beta_c V_2-\frac{\beta_c^2}{2}}(M'(\beta))^{(2)},
\end{equation}

\noindent with similar independence structure.

\vspace{0.3cm}

Finally, for $\beta>\beta_c$, $n^{\frac{3\beta}{2\beta_c}}e^{\frac{n}{2}(\beta-\beta_c)^2}M_n(\beta)$ converges in law to a non-trivial random variable $\widetilde{M}(\beta)$ satisfying

\begin{equation}
\widetilde{M}(\beta)\stackrel{d}{=}\frac{1}{2}e^{\beta V_1-\beta\beta_c+\frac{\beta_c^2}{2}}\widetilde{M}^{(1)}(\beta)+\frac{1}{2}e^{\beta V_2-\beta\beta_c+\frac{\beta_c^2}{2}}\widetilde{M}^{(2)}(\beta).
\end{equation}

Moreover, the law of $\widetilde{M}(\beta)$ can be written as

\begin{equation}
\widetilde{M}(\beta)\stackrel{d}{=}L_{\frac{\beta_c}{\beta}}(M'(\beta_c)),
\end{equation}

\noindent where the process $t\mapsto L_\alpha(t)$, $t\geq 0$ is a stable L\'evy subordinator of index $\alpha$, independent of $M'(\beta_c)$.
\end{theorem}

\begin{remark} As we shall soon see, $Z_{n,\beta}$ can be viewed as the partition function of a type of Random Energy Model with specially (logarithmically) correlated energies and the fact that the law of $\widetilde{M}(\beta)$ is essentially characterized by $M'(\beta_c)$ is one aspect of the freezing transition occurring in the study of $\log$-correlated disordered systems, see  \cite{cld,FyoRev2010,fg} for discussions of various physical aspects of freezing and its relation to multifractality, and further references.
\end{remark}

\begin{remark} For this theorem, the relevant property of the stable subordinator is that it satisfies $\E(e^{-sL_\alpha(t)})=e^{-Cts^\alpha}$ for some positive constant $C$.
\end{remark}

In addition to proving convergence, the following properties of the limiting random variables are known:

\begin{theorem}[\cite{tailsubc,tailc,dl,molch}]\label{th:cascmom}
For $\beta<\beta_c$, $\E(M(\beta)^p)<\infty$ if and only if $p<\frac{\beta_c^2}{\beta^2}$ and $\E(M'(\beta_c)^p)<\infty$ if and only if $p<1$. In fact, one has the sharper result that there exists a positive constant $c(\beta)$ so that for $\beta<\beta_c$

\begin{equation}\label{momtailhightemp}
\lim_{x\to\infty}x^{\frac{\beta_c^2}{\beta^2}}\mathbb{P}(M(\beta)\geq x)=c(\beta)
\end{equation}

\noindent and

\begin{equation}\label{momtaillowtemp}
\lim_{x\to\infty}x\mathbb{P}(M'(\beta_c)\geq x)=c(\beta_c).
\end{equation}

\end{theorem}

\begin{remark} For $\beta>\beta_c$, the tail asymptotics $\mathbb{P}(\widetilde{M}(\beta)\geq x)\sim x^{-\frac{\beta_c}{\beta}}$ follow from the representation $\widetilde{M}(\beta)\stackrel{d}{=}L_{\frac{\beta_c}{\beta}}(M'(\beta_c))$ and the corresponding asymptotics of $L_{\frac{\beta_c}{\beta}}(1)$ and $M'(\beta_c)$.
\end{remark}

\begin{figure}
\centering
\includegraphics[trim=6cm 17cm 2cm 3cm, scale=0.7]{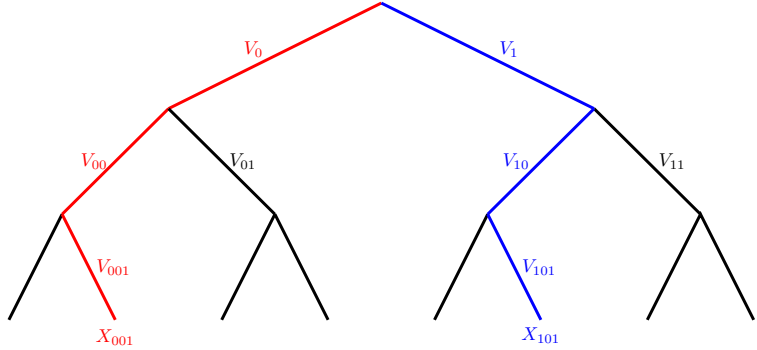}
\caption{A graphical representation of the binary tree and the related random variables used to construct a branching random walk. Here $X_{001}=V_0+V_{00}+V_{001}$ and $X_{101}=V_1+V_{10}+V_{101}$.}
\end{figure}

\vspace{0.3cm}

Let us now turn to a few words on the tools typically used for studying the recursion. Perhaps the most efficient (and most easily generalized) method of studying $Z_n(\beta)$ is through a branching random walk representation. Consider a binary tree and on each edge, place i.i.d. copies of standard normal variables. Each edge at hight $n$ in the tree can be labelled by a binary sequence $\sigma=(\sigma_1,...,\sigma_n)\in\lbrace 0,1\rbrace^n$, where $0$ refers to taking the left path and $1$ to taking the right one at a branching point. Let us denote by $V_\sigma$ the variable associated to the edge $\sigma$ and if $\sigma\in\lbrace 0,1\rbrace^n$ and $k\leq n$, write $\sigma|k$ for the truncation of $\sigma$ to length $k$: $\sigma|k=(\sigma_1,...,\sigma_k)$.

\vspace{0.3cm}

Let us truncate the tree at height $n$. Each branch in the truncated tree can again be identified with a binary sequence $\sigma=(\sigma_1,...,\sigma_n)\in\lbrace 0,1\rbrace^n$. Call $X_\sigma$ the sum of the $V_i$ along the branch: $X_{\sigma}=\sum_{k=1}^n V_{\sigma|k}$.  The collection $(X_\sigma)$ can be viewed as a branching random walk. See Figure 1 for an illustration of the construction. Now let us define

\begin{equation}
Z_n(\beta)=\sum_{\sigma\in\lbrace 0,1\rbrace^n}e^{\beta X_\sigma}.
\end{equation}

\begin{remark} Note that if we think of $-X_\sigma$ as the (random) energy of a spin-configuration $\sigma$, $Z_n(\beta)$ becomes the partition function of a kind of Random Energy Model with correlations. In fact, one can check that the covariance of $X_\sigma$ and $X_{\sigma'}$ is the logarithm of the distance of $\sigma$ and $\sigma'$ in the dyadic metric of the tree. Partition functions of Random Energy models with logarithmic correlations in Euclidean metric have attracted recently considerable attention due to several interesting applications in Statistical Mechanics, Random Matrix Theory and Number Theory see e.g. \cite{FB08,FLDR2009,FLDR2012,FyoKeat2014,ABH} and references therein. They share many characteristic properties with their diadic tree counterparts, such e.g. as the tail behaviour of moments (\ref{momtailhightemp}). Analysis of such models frequently proceeds by identifying an approximate binary tree structure, see e.g. \cite{ABH}.

\end{remark}

For $\sigma\in\lbrace 0,1\rbrace^k$, we also define a similar object corresponding to the tree of height $n$ rooted at $\sigma$:

\begin{equation}
Z_n^{(\sigma)}(\beta)=\sum_{\stackrel{\sigma'\in \lbrace 0,1\rbrace^{n+k}:}{\sigma'|k=\sigma}}e^{\beta (X_{\sigma'}-X_{\sigma})}.
\end{equation}

Note that $Z_n^{(\sigma)}(\beta)\stackrel{d}{=}Z_n(\beta)$ and for $\sigma,\sigma'\in \lbrace 0,1\rbrace^k$ and $\sigma\neq \sigma'$, $Z_n^{(\sigma)}(\beta)$ is independent from $Z_n^{(\sigma')}(\beta)$ and $X_{\sigma}$ is independent from $Z_n^{(\sigma')}(\beta)$ for any $\sigma,\sigma'\in\lbrace 0,1\rbrace^k$. Also one has (simply plugging in the definitions)

\begin{equation}\label{eq:casc2}
Z_{n+k}(\beta)=\sum_{\sigma\in \lbrace 0,1\rbrace^k}e^{\beta X_\sigma}Z_n^{(\sigma)}(\beta).
\end{equation}

In particular, specializing to $k=1$,

\begin{equation}
Z_{n+1}(\beta)=e^{\beta V_0}Z_n^{(0)}(\beta)+e^{\beta V_1}Z_n^{(1)}(\beta),
\end{equation}

\noindent so we have realized the variables appearing in our recursion on the same probability space. Moreover, if we denote by $\mathcal{F}_n$ the $\sigma$-algebra generated by the $V_\sigma$ with $\sigma\in\lbrace 0,1\rbrace^k$ for any $k\leq n$, then it follows directly from \eqref{eq:casc2} (by setting $n=1$) that $M_n(\beta)=\frac{Z_n(\beta)}{\E(Z_n(\beta))}$ is a positive martingale and one sees immediately that it converges to some limit. Then analyzing moments of $M_n(\beta)$ through the recursion, one finds that the limit is non-trivial for $\beta<\beta_c$ and zero for $\beta\geq \beta_c$. Our analysis of the moments of $P_n(q)$ will be very similar to the way one can analyze the moments here so we shall not discuss this further now.

We do however point out that while in the model we have described, the limiting object $M(\beta)$ has moments only up to order $p<\beta_c^2/\beta^2$, Barral and Mandelbrot \cite{bm} considered a slightly different kind of model where $e^{\beta V_i}$ were replaced by variables that could take negative values as well, and in this case one has a limiting random variable which has moments of all positive integer order. Moreover, they used the values of the limiting moments to identify the law of the limiting random variable. This is similar to our argument for small $q$.

\vspace{0.3cm}

Finding the correct normalization for $\beta\geq \beta_c$ is far more complicated. At $\beta_c$ what turns out to be critical is proving that the so-called derivative martingale (see \cite{bk})

\begin{equation}
-M_n'(\beta_c)=2^{-n}\sum_{\sigma\in \lbrace 0,1\rbrace^n}(n\beta_c-X_\sigma)e^{\beta_c X_\sigma-\frac{\beta_c^2}{2}n}
\end{equation}

\noindent converges to an almost surely positive random variable. In fact, the limit equals (in distribution) $M'(\beta_c)$ up to a constant multiple.

\vspace{0.3cm}

For $\beta>\beta_c$, what turns out to be the important question is the behavior of $\max_{\sigma\in \lbrace 0,1\rbrace^n}X_\sigma$. What is typical of the freezing transition is that for large $\beta$, only the variables $X_\sigma$ which are close to the maximum one (on that level in the tree) contribute to $Z_n(\beta)$. The result for the maximum is the following:

\begin{theorem}[\cite{aid,cw}]\label{th:cascmax}
There exists a positive constant $c>0$ such that

\begin{align}\label{extreme}
\lim_{n\to\infty}\mathbb{P}&\left(\max_{\sigma\in\lbrace 0,1\rbrace^n}X_\sigma-\sqrt{2\log 2}n+\frac{3}{2\sqrt{2\log 2}}\log n\leq x\right)\\ \nonumber
&=\E(\exp(-c e^{-\beta_c x}M'(\beta_c))).
\end{align}

\end{theorem}

\begin{remark} One feature of this limiting distribution is that it has a tail of the formm $xe^{-\beta_c x}$ (as $x\to \infty)$ so this is a problem of extreme value statistics that is not in the Gumbel universality class. As conjectured in the physical literature \cite{cld,FLDR2009} such a tail together with the coefficient $3/2$ in (\ref{extreme}) are characteristic features of the generic sequences, processes and fields with logarithmic correlations, and are intimately related to the powerlaw tail of the partition function moments (\ref{momtailhightemp}), see \cite{FLDR2012}. Rigorous results corroborating  this conjecture can be found e.g. in \cite{DRZ2015} and references therein.

\end{remark}

The proof of Theorem \ref{th:cascmax} is quite technical, but a fundamental idea is proving that branches $X_\sigma$ that live on the scale of the maximum end up there along paths which perform an excursion around a linear path of speed $\sqrt{2\log 2}$.

\section{Scaling of moments in the LME recursion}

Compared to the branching random walk, there is an immediate difficulty in analyzing the convergence of the solution to the LME recursion (\ref{LME}). Namely even if one realizes the random variables on the same probability space in tree form, one does not have a martingale structure. We have not noticed any monotonicity or really any contractivity either that would enable proving convergence in general. Thus instead of focusing on proving convergence, we will begin by studying the moments of $P_n(q)$ satisfying the LME recursion (\ref{LME}). Using arguments similar to ones used for the branching random walk, we will see that there is a phase transition: there is an explicit $q_c$ (approximately $2.4056$) such that for $\frac{1}{2}<q<q_c$, for suitable $p>1$, $\E(P_n(q)^p)\sim \E(P_n(q))^p$ while for $q>q_c$, $\frac{P_n(q)}{\E(P_n(q))}\stackrel{d}{\to}0$. In the following sections we will also use these results to show that for $\frac{1}{2}<q<1$, $\frac{P_n(q)}{\E(P_n(q))}$ converges to a non-trivial random variable. For brevity, let us introduce a name for this random variable.

\begin{definition}
Let $P_n(q)$ be given by the LME-recursion (\ref{LME}) and let $q>\frac{1}{2}$. We then write
\begin{equation}
\Pi_n(q)=\frac{P_n(q)}{\E(P_n(q))}.
\end{equation}
\end{definition}

\vspace{0.3cm}

Note that if one is able to prove that $\sup_n\E(\Pi_n(q)^p)<\infty$ for some $p>1$, then one has immediately that there exist constants $0<C_1,C_2<1$ such that for all $n$,

\begin{equation}
\begin{array}{ccc}
\mathbb{P}(\Pi_n(q)>2)\leq C_1 & \mathrm{and} & \mathbb{P}\left(\Pi_n(q)>\frac{1}{2}\right)\geq C_2.
\end{array}
\end{equation}

The first inequality is simply Markov's inequality, since $\E(\Pi_n(q))=1$. The second inequality is a Paley-Zygmund type inequality: we have by H\"older's inequality

\begin{align*}
1&=\E\left(\Pi_n(q)\mathbf{1}\left\lbrace \Pi_n(q)\geq \frac{1}{2}\right\rbrace\right)+\E\left(\Pi_n(q)\mathbf{1}\left\lbrace \Pi_n(q)\leq \frac{1}{2}\right\rbrace\right)\\
&\leq (\E(\Pi_n(q)^p))^{\frac{1}{p}}\left(\mathbb{P}\left(\Pi_n(q)\geq \frac{1}{2}\right)\right)^{1-\frac{1}{p}}+\frac{1}{2}
\end{align*}

\noindent yielding

\begin{equation}
\mathbb{P}\left(\Pi_n(q)\geq \frac{1}{2}\right)\geq \left(\frac{1}{2(\E(\Pi_n(q)^p))^{\frac{1}{p}}}\right)^{\frac{p}{p-1}}.
\end{equation}

So in particular, the scaling properties of the moments imply that we can't have $\Pi_n(q)\approx \infty$ with high probability or $\Pi_n(q)\approx 0$ with high probability. Thus for such $q$, $\E(P_n(q))$ is a scale (though we expect it's the only one) at which $P_n(q)$ lives.

\vspace{0.3cm}

For our proof, it will be important to understand the asymptotics of $\E(P_n(q))$. To do this, we consider $T_n(q):=\frac{\E(P_{n+1}(q))}{\E(P_n(q))}=\E(\sin^{2q}\theta_n+\cos^{2q}\theta_n)$. The central result for this quantity is

\begin{lemma}\label{le:estT}
For $q>\frac{1}{2}$, we have for $\alpha=\min(2,2q)$

\begin{equation}
1-T_n(q)=\E(1-\sin^{2q}\theta_n-\cos^{2q}\theta_n)=\frac{b}{n}T(q)+\mathcal{O}\left(\left(\frac{b}{n}\right)^{\alpha}\right),
\end{equation}

\noindent where

\begin{align}
T(q)%&=\frac{2}{\pi}\int_0^{\frac{\pi}{2}}\frac{d\theta}{\cos^2\theta\sin^2\theta}(1-\sin^{2q}\theta-\cos^{2q}\theta)\notag\\
&=\frac{\sqrt{\pi}}{8}\frac{\Gamma\left(q-\frac{1}{2}\right)}{\Gamma(q-1)}.
\end{align}
\end{lemma}

\begin{remark}
 Note that $T(q)$ is finite for $q>\frac{1}{2}$.
\end{remark}

\begin{remark} We choose to write the errors in terms of $\frac{b}{n}$ instead of $\frac{1}{n}$ as one might potentially be interested in making $b$ depend on $n$.
\end{remark}

\begin{proof}

Note that as $\theta\to 0$, $1-\sin^{2q}\theta-\cos^{2q}\theta=\mathcal{O}(|\theta|^{\min(2q,2)})$, so by Lemma \ref{sing},

\begin{equation}
1-T_n(q)=\frac{b}{n}\int_{-\frac{\pi}{4}}^{\frac{\pi}{4}}\frac{d\theta}{\sin^{2}2\theta}(1-\sin^{2q}\theta-\cos^{2q}\theta)+\mathcal{O}((bn^{-1})^{\alpha}).
\end{equation}

Recalling our convention for $\sin^{2q}\theta$ and $\cos^{2q}\theta$, we see that the integrand is symmetric under $\theta\to-\theta$. We also note that for $\theta\in[0,\pi/4]$, $(1-\cos^{2q}\theta-\sin^{2q}\theta)/\sin^22\theta$ is symmetric under $\theta\to -\theta+\pi/2$ so we find that

\begin{equation}
T(q)=\frac{1}{4}\int_0^{\frac{\pi}{2}}\frac{d\theta}{\sin^2\theta\cos^2\theta}(1-\sin^{2q}\theta-\cos^{2q}\theta).
\end{equation}

To see that the integral can indeed be written as a ratio of $\Gamma$-functions, let us introduce an auxiliary complex parameter to the integral:

\begin{equation}
T(q,z):=\frac{1}{4}\int_{0}^{\frac{\pi}{2}}(1-\sin^{2q}\theta-\cos^{2q}\theta)(\sin \theta \cos\theta)^{z-2}d\theta.
\end{equation}

We see that $T(q,0)=T(q)$ and $T(q,z)$ is analytic in $z$ for $\mathrm{Re}(z)>\max(1-2q,-1)$. For $\mathrm{Re}(z)$ large enough, we see that each term in the integral can be expressed as a Beta-function, which we can write as a ratio of $\Gamma$-functions: for $\mathrm{Re}(z)>1$,

\begin{equation}
T(q,z)=\frac{1}{8}\frac{\Gamma(\frac{z-1}{2})^2}{\Gamma(z-1)}-\frac{1}{4}\frac{\Gamma(q+\frac{z-1}{2})\Gamma(\frac{z-1}{2})}{\Gamma(q+z-1)}.
\end{equation}

As the only singularities of the $\Gamma$-function in the complex plane are simple poles at non-positive integers, we see that for $\mathrm{Re}(z)>\max(1-2q,-1)$, the only possible singularity of the meromorphic function

\begin{equation}
z\mapsto\frac{1}{8}\frac{\Gamma(\frac{z-1}{2})^2}{\Gamma(z-1)}-\frac{1}{4}\frac{\Gamma(q+\frac{z-1}{2})\Gamma(\frac{z-1}{2})}{\Gamma(q+z-1)}
\end{equation}

\noindent is at $z=1$, but we see that at $z=1$, the poles of the two different terms cancel and

\begin{equation}
T(q,z)=\frac{1}{8}\frac{\Gamma(\frac{z-1}{2})^2}{\Gamma(z-1)}-\frac{1}{4}\frac{\Gamma(q+\frac{z-1}{2})\Gamma(\frac{z-1}{2})}{\Gamma(q+z-1)}
\end{equation}

\noindent even for $\mathrm{Re}(z)>\max(1-2q,-1)$ . In particular,

\begin{equation}
T(q)=T(q,0)=-\frac{1}{4}\frac{\Gamma(q-\frac{1}{2})\Gamma(-\frac{1}{2})}{\Gamma(q-1)}=\frac{\sqrt{\pi}}{8}\frac{\Gamma(q-\frac{1}{2})}{\Gamma(q-1)}.
\end{equation}

\end{proof}

We shall also need similar estimates for joint moments of $\sin\theta_n$ and $\cos\theta_n$.

\begin{lemma}\label{le:estB}
For $q>\frac{1}{2}$ and $j,k\geq 1$ we have for $\alpha=\min(2qj,2qk)>1$
\begin{equation}
\E(\sin^{2qj}\theta_n \cos^{2qk}\theta_n)=\frac{b}{n}\frac{1}{2}\int_0^{\frac{\pi}{4}}\sin^{2qj-2}\theta \cos^{2qk-2}\theta d\theta+\mathcal{O}\left(\left(\frac{b}{n}\right)^{\alpha}\right).
\end{equation}

\end{lemma}

\begin{proof}
Noting that as $\theta\to 0$, $\sin^{2qj}\theta \cos^{2qk}\theta=\mathcal{O}(|\theta|^\alpha)$, we again refer to Lemma \ref{sing}. We also make use of the symmetry under $\theta\to-\theta$ (hidden in our notational convention).
\end{proof}

\vspace{0.3cm}

Before moving on to studying the moments of $P_n$, we fix some notation.

\begin{lemma}
There is only a single solution (which we call $q_c$ from now on) to the equation $qT'(q)=T(q)$ in the domain $(\frac{1}{2},\infty)$.
\end{lemma}

\begin{proof}
As this is equivalent to showing that the function $H(q)=qT'(q)-T(q)$ has only a single zero, it is enough to check that $H'(q)=qT''(q)$ does not change its sign and that $H(q)$ does. The first claim is clear for example from the representation

\begin{equation}
T(q)=\frac{1}{4}\int_0^{\frac{\pi}{2}}\frac{d\theta}{\sin^2\theta\cos^2\theta}(1-\sin^{2q}{\theta}-\cos^{2q}\theta).
\end{equation}

For the second one we note that as $T(q)<0$ for $q<1$ and $T'(q)>0$ for all $q$, we see that $qT'(q)-T(q)>0$ for small $q$. For large $q$, one can use the asymptotics of the polygamma function of order 0 to check that $H(q)<0$ for large $q$ (in fact, one can check that $H(q)\sim -\frac{1}{2}\sqrt{q}$ for large $q$). Moreover, numerically one finds that the root $q_c$ is at 2.4056.
\end{proof}

We also need yet another elementary remark to describe which moments of $\Pi_n(q)$ are bounded

\begin{lemma}
For $\frac{1}{2}<q<q_c$, the set  $\lbrace p>1:T(pq)-pT(q)>0\rbrace$ is an open interval.
\end{lemma}

\begin{proof}
For $g(p):=T(pq)-pT(q)$, $g'(p)=qT'(pq)-T(q)$ and $g''(p)=q^2 T''(pq)$. If the set was not an interval, there would be several points at which $g'(p)=0$. But since $g''$ is strictly negative, $g'$ is strictly decreasing, we see that this is impossible. The interval is non-empty since for $q<q_c$, $g'(1)>0$.
\end{proof}

\begin{remark} For $q<q_c$, we will write $\lbrace p>1: T(pq)-pT(q)>0\rbrace$ as $(1,p^*(q))$. Moreover, as $\partial_q (T(pq)-pT(q))=p(T'(pq)-T'(q))<0$ for $p>1$ (as $T''(q)<0$), we see that $p^*(q)$ is decreasing in $q$.
\end{remark}

We are now in a position to characterize which moments are finite:

\begin{proposition}
For $\frac{1}{2}<q<q_c$, $\sup_n \E(\Pi_n(q)^p)<\infty$ for $1<p<p^*(q)$. For $p>p^*(q)$, $\E(\Pi_n(q)^p)\to \infty$ as $n\to \infty$.
\end{proposition}

\begin{proof}
Let us begin by considering $p<p^*(q)$. By Jensen's inequality, for $p'<p$, $\E(\Pi_n(q)^{p'})=\E(\Pi_n(q)^{p\frac{p'}{p}})\leq (\E(\Pi_n(q)^p))^{\frac{p'}{p}}$. Thus it is enough to consider $p$ close enough to $p^*(q)$. Let us further assume that $p\leq 2$ for now and use the subadditivity of $x\mapsto x^{\frac{p}{2}}$ to see that

\begin{align*}
\E(P_{n+1}(q)^p)&=\E\left((\sin^{2q}\theta_nP_n^{(1)}(q)+\cos^{2q}\theta_nP_n^{(2)}(q))^{2\frac{p}{2}}\right)\\
&\leq \E(\sin^{2qp}\theta_n P_n^{(1)}(q)^p+2\sin^{qp}\theta_n\cos^{qp}\theta_nP_n^{(1)}(q)^{\frac{p}{2}}P_n^{(2)}(q)^{\frac{p}{2}})\\
&+\E(\cos^{2qp}\theta_n P_n^{(2)}(q)^p).
\end{align*}

Normalizing by $\E(P_{n+1}(q))^p$ and noting that by Jensen, $\E(P_n(q)^{\frac{p}{2}})\leq \E(P_n(q))^{\frac{p}{2}}$, we see that

\begin{equation}
\E(\Pi_{n+1}(q)^p)\leq \frac{\E(\sin^{2qp}\theta_n+\cos^{2qp}\theta_n)}{(\E(\sin^{2q}\theta_n+\cos^{2q}\theta_n))^p}\E(\Pi_n(q)^p)+\frac{\E(2\sin^{qp}\theta_n\cos^{qp}\theta_n)}{(\E(\sin^{2q}\theta_n+\cos^{2q}\theta_n))^p}.
\end{equation}

Let us now assume that $\E(\Pi_{n+1}(q)^p)\geq \E(\Pi_n(q)^p)$. We then obtain that

\begin{equation}
\frac{n}{b}\left(1-\frac{\E(\sin^{2qp}\theta_n+\cos^{2qp}\theta_n)}{(\E(\sin^{2q}\theta_n+\cos^{2q}\theta_n))^p}\right)\E(\Pi_{n+1}(q)^p)\leq\frac{n}{b}\frac{\E(2\sin^{qp}\theta_n\cos^{qp}\theta_n)}{(\E(\sin^{2q}\theta_n+\cos^{2q}\theta_n))^p}.
\end{equation}
As $n\to \infty$, the term multiplying $\E(\Pi_ {n+1}(q)^p)$ converges to $T(pq)-pT(q)$ which is positive by our assumption. If $q>1$, the right side converges (by Lemma \ref{le:estB}) to

\begin{equation}
\int_0^{\frac{\pi}{4}}\sin^{qp-2}\theta\cos^{qp-2}\theta d\theta.
\end{equation}

\noindent and this integral converges for any $p>1$. If $q<1$, we note $T(pq)-pT(q)>0$ for any $p>1$, in particular choosing $p=2$, we can again apply Lemma \ref{le:estB}, and we get a convergent integral as $q>\frac{1}{2}$. We conclude that if $\E(\Pi_{n+1}(q)^p)\geq \E(\Pi_n(q)^p)$, we get an absolute bound for $\E(\Pi_{n+1}(q)^p)$. On the other hand, if $\E(\Pi_{n+1}(q)^p)\leq \E(\Pi_n(q)^p)$, we keep going back in the sequence until we have an index of increase or we reach $n=1$ - in any event, we get an upper bound. We conclude that $\E(\Pi_n(q)^p)$ is bounded for $p\in[1,2]$.

\vspace{0.3cm}

One can then proceed inductively to prove this if $p\in(k-1,k]$, for $k>2$: one writes $p=k\frac{p}{k}$ and with a similar subadditivity argument one has

\begin{align}
\E(\Pi_{n+1}(q)^p)&\leq\frac{1}{T_n(q)^p}\sum_{j=0}^k{k\choose j}\E\left(\sin^{2qp\frac{j}{k}}\theta_n \cos^{2qp\frac{k-j}{k}}\theta_n\right)\\
& \times\E(\Pi_n(q))^\frac{pj}{k}) \E(\Pi_n(q)^{\frac{p(k-j)}{k}}). \notag
\end{align}

Assuming $\E(\Pi_{n+1}(q)^p)\geq \E(\Pi_n(q)^p)$ gives

\begin{align}
\E(\Pi_{n+1}(q)^p)\left(1-\frac{T_n(qp)}{T_n(q)^p}\right)&\leq \frac{1}{T_n(q)^p}\sum_{j=1}^{k-1}{k\choose j}\E\left(\sin^{2qp\frac{j}{k}}\theta_n \cos^{2qp\frac{k-j}{k}}\theta_n\right)\\
& \times\E(\Pi_n(q))^\frac{pj}{k}) \E(\Pi_n(q)^{\frac{p(k-j)}{k}}).\notag
\end{align}

\vspace{0.3cm}

We now assume that we know that $\sup_n\E(\Pi_n(q)^{p'})<\infty$ for $1<p'\leq k-1$. As $p\leq k$, we can again use Jensen to estimate $\E(\Pi_n(q)^{\frac{pj}{k}})\leq \E(\Pi_n(q)^j)^{\frac{p}{k}}$ which is bounded in $n$ for $1<j<k$ (by our induction hypothesis). We can then proceed as for $p\leq 2$.

\vspace{0.3cm}

Consider now $p>p^*(q)$. Then by superadditivity of $x\mapsto x^p$ for $p>1$, we have

\begin{equation}
\E(\Pi_{n+1}(q)^p)\geq \frac{\E(\sin^{2qp}\theta_n+\cos^{2qp}\theta_n)}{(\E(\sin^{2q}\theta_n+\cos^{2q}\theta_n))^p}\E(\Pi_n(q)^p).
\end{equation}

The prefactor is $1+\frac{b}{n}(pT(q)-T(pq))+\mathcal{O}((\frac{b}{n})^{1+\epsilon})$ for some $\epsilon>0$ so iterating this, we find that $\E(\Pi_n(q)^p)\geq C e^{b\log n(pT(q)-T(pq))}$. By our assumptions, the exponent is positive so we see that this blows up.

\end{proof}

\begin{remark}[Multifractal scaling]
We point out that this result implies the multifractal scaling we discussed earlier (see \eqref{eq:scaling}): for $q<q_c$, $\E(\Pi_n(q)^p)$ is bounded for some $p>1$ and we argued that this implies that asymptotically, $\Pi_n(q)$ is non-trivial. In fact, we expect that it converges, so we'll make the assumption that as $n\to \infty$, $\Pi_n(q)=\mathcal{O}(1)$ (we'll see that for $q<1$, we have convergence so at least in this case, our argument is solid). We thus have 

\begin{align}
\notag P_n(q)&=\Pi_n(q)\E(P_n(q))\\
&=\mathcal{O}(1)\prod_{k=1}^n \left(1-\frac{b}{k}T(q)\right)\\
\notag &=\mathcal{O}(1)n^{-bT(q)}.
\end{align}

In the notation of \eqref{eq:scaling}, this implies 

\begin{equation}
d(q)=\frac{bT(q)}{q-1}=b\frac{\sqrt{\pi}}{8}\frac{\Gamma(q-\frac{1}{2})}{\Gamma(q)}.
\end{equation}

If we fix some $q>1/2$, by taking $b$ small enough, we have $d(q)\in(0,1)$ implying the claimed multifractal scaling for $P_n(q)$. For larger values of $q$, we expect a different normalizing factor, but still one that produces multifractality.
\end{remark}

To see that $q_c$ described here is truly the critical point in that for $q>q_c$, normalizing by the expectation is not the correct way to proceed, we point out the following result:

\begin{proposition}
For $q>q_c$, there exists an $h=h(q)<1$ such that

\begin{equation}
\E(\Pi_n(q)^h)\to 0
\end{equation}

\noindent as $n\to \infty$.
\end{proposition}

\begin{proof}
This follows essentially from the subadditivity of $x\mapsto x^h$. We get immediately that

\begin{equation}
\E(\Pi_{n+1}(q)^h)\leq \E(\Pi_n(q)^h)\frac{\E(\sin^{2qh}\theta_n+\cos^{2qh}\theta_n)}{(\E(\sin^{2q}\theta_n+\cos^{2q}\theta_n))^h}.
\end{equation}

The ratio here is $1-\frac{b}{n}(T(qh)-hT(q))+\mathcal{O}((\frac{b}{n})^{1+\epsilon})$ (for some $\epsilon>0$). As $qT'(q)-T(q)<0$, One can check that for $h\in(0,1)$ close enough to one, $T(qh)-hT(q)>0$. Thus

\begin{equation}
\E(\Pi_n(q)^h)\leq C e^{-b\log n(T(qh)-hT(q))}\to 0
\end{equation}

\noindent as $n\to \infty$.
\end{proof}

\section{Convergence of positive integer moments}

In this section we will prove that for each positive integer $k$ and for $q\in (\frac{1}{2},q_k)$, where $q_k$ is the unique point satisfying $T(kq_k)-kT(q_k)=0$,  $\E(\Pi_n(q)^j)$ converges for $j\leq k$.

\begin{proposition}
For any fixed positive value of $b$, $\frac{1}{2}<q<q_k$ and $j\leq k$, $\E(\Pi_n(q)^j)$ converges to a positive number $M_j=M_j(q)$ and for $l\leq k$ one has the recursion

\begin{equation}\label{eq:momrec}
M_l=\frac{1}{T(lq)-lT(q)}\sum_{j=1}^{l-1}{l\choose j} M_j M_{l-j}\frac{1}{2}\int_0^{\frac{\pi}{4}}\sin^{2qj-2}\theta\cos^{2q(l-j)-2}\theta d\theta
\end{equation}

\noindent with $M_1(q)=1$.

\end{proposition}

\begin{proof}
Let us begin with the $k=2$ case. Normalizing the recursion $\eqref{lemrec}$ by expectations, we find that the second moment satisfies the recursion

\begin{align*}
\E(\Pi_{n+1}(q)^2)&=\E(\Pi_n(q)^2)+\left(\frac{T_n(2q)}{T_n(q)^2}-1\right)\E(\Pi_n(q)^2)\\
&+\frac{2}{T_n(q)^2}\E(\sin^{2q}\theta_n\cos^{2q}\theta_n).
\end{align*}

By Lemma \ref{le:estT} and Lemma \ref{le:estB}, we can write (for some $\alpha>1$)

\begin{align}
\frac{n}{b}(\E(\Pi_{n+1}(q)^2-\E(\Pi_n(q)^2))&=\int_0^{\frac{\pi}{4}}\sin^{2q-2}\theta\cos^{2q-2}\theta d\theta\\
&-(T(2q)-2T(q))\E(\Pi_n(q)^2)+\mathcal{O}\left(\left(\frac{b}{n}\right)^{\alpha-1}\right).\notag
\end{align}

From this we can see what the limit of $\E(\Pi_n(q)^2)$ must be if it exists. If $\E(\Pi_n(q)^2)$ were to converge, the left side can only converge to zero - if it had a non-zero limit, $\E(\Pi_n(q)^2)=\sum_{k=2}^{n}(\E(\Pi_k(q)^2)-\E(\Pi_{k-1}(q)^2))+\E(\Pi_1(q)^2)$ would diverge (logarithmically) as $n\to\infty$. Thus if a limit exists it must equal $M_2(q)=\frac{1}{T(2q)-2T(q)}\int_0^{\frac{\pi}{4}}\sin^{2q-2}\theta\cos^{2q-2}\theta d\theta$ (recall the notation of \eqref{eq:momrec}).

\vspace{0.3cm}

Let us thus write

\begin{equation}
\E(\Pi_n(q)^2)=M_2(q)-\epsilon_n^{(2)}.
\end{equation}

A priori, we only know that $\epsilon_n^{(2)}$ is a bounded sequence (as we showed in the previous section that $\E(\Pi_n(q)^2)$ is bounded for $\frac{1}{2}<q<q_2$), but let us try to show that $\epsilon_n^{(2)}\to 0$. Plugging this ansatz into the recursion, one finds
\begin{equation}
\epsilon_{n+1}^{(2)}=\epsilon_n^{(2)}\left(1-\frac{b}{n}(T(2q)-2T(q))\right)+\delta_n\notag,
\end{equation}

\noindent where $\delta_n=\mathcal{O}((\frac{b}{n})^{\alpha})$. We can of course write the solution to this recursion "explicitly": if we define $\delta_0=\epsilon_1^{(2)}$ and interpret the empty product ($k=n-1$) as $1$, then

\begin{equation}
\epsilon_n^{(2)}=\sum_{k=0}^{n-1}\delta_k\prod_{j=k+1}^{n-1}\left(1-\frac{b}{j}(T(2q)-2T(q))\right).
\end{equation}

The product we can estimate by

\begin{equation}
\left|\prod_{j=k+1}^{n-1}\left(1-\frac{b}{j}(T(2q)-2T(q))\right)\right|\leq C \left(\frac{k}{n}\right)^{b(T(2q)-2T(q))},
\end{equation}

\noindent where $C$ is independent of $k$ and $n$. Note that if $b$ is large (but independent of $n$), the terms in the product may be negative for small $j$, but this still is a valid asymptotic estimate for large $n$. Recalling that $|\delta_k|\leq Ck^{-\alpha}$, for some $\alpha>1$, we see that

\begin{align}
|\epsilon_n^{(2)}|&\leq Cn^{-\alpha}+Cn^{-b(T(2q)-2T(q))}\sum_{k=1}^{n-2} k^{-\alpha+b(T(2q)-2T(q))}\\
\notag &\leq \begin{cases}
Cn^{-b(T(2q)-2T(q))}, & -\alpha+b(T(2q)-2T(q))<-1\\
Cn^{-b(T(2q)-2T(q))}\log n, & -\alpha+b(T(2q)-2T(q))=-1\\
Cn^{1-\alpha}, & -\alpha+b(T(2q)-2T(q))>-1
\end{cases}.
\end{align}

 Thus we see that for any fixed value of $b$, $\lim_{n\to\infty}\epsilon_n^{(2)}=0$ implying that $\E(\Pi_n(q)^2)\to M_2(q)$.

\vspace{0.3cm}

Let us now try to do something similar for $\E(\Pi_n(q)^k)$. From the recursion for $\Pi_n(q)$, one finds

\begin{align*}
&\frac{n}{b}\left(\E(\Pi_{n+1}(q)^k)-\E(\Pi_n(q)^k)\right)\\
&=-(T(qk)-kT(q)) \E(\Pi_n(q)^k)+\sum_{j=1}^{k-1} {k\choose j}\E(\Pi_n(q)^j)\E(\Pi_n(q)^{k-j})\\
& \times\frac{1}{2}\int_0^{\frac{\pi}{4}}\frac{d\theta}{\sin^2\theta\cos^2\theta}\sin^{2qj}\theta\cos^{2q(k-j)}\theta+\mathcal{O}\left(\left(\frac{b}{n}\right)^{\alpha-1}\right).
\end{align*}

With a similar argument as in the $k=2$ case, we see that if $\E(\Pi_n(q)^j)$ were to converge for $j\leq k$, then the only possibility is to $M_j(q)$ which satisfies the recursion

\begin{align}
M_k(q)=&\frac{1}{T(qk)-kT(q)}\sum_{j=1}^{k-1} {k\choose j}M_j(q)M_{k-j}(q)\\
& \times\frac{1}{2}\int_0^{\frac{\pi}{4}}\frac{d\theta}{\sin^2\theta\cos^2\theta}\sin^{2qj}\theta\cos^{2q(k-j)}\theta\notag
\end{align}

\noindent with $M_1(q)=1$.

\vspace{0.3cm}

Based on our experience in the $k=2$ case, let us write

\begin{equation}
\E(\Pi_n(q)^j)=M_j(q)-\epsilon_n^{(j)}
\end{equation}

\noindent and proceed inductively: let us assume that for $j<k$, $|\epsilon_m^{(j)}|\leq C_j m^{-\alpha_j}$ for some $\alpha_j>0$ independent of $m$. Plugging this into the recursion for $\E(\Pi_n(q)^k)$ (and making use of Lemma \ref{le:estT} and Lemma \ref{le:estB}), we find

\begin{equation}
\epsilon_{n+1}^{(k)}=\epsilon_n^{(k)}\left(1-\frac{b}{n}(T(qk)-kT(q))\right)+\mathcal{O}(n^{-\alpha})
\end{equation}

\noindent for some $\alpha>1$. Arguing as for $k=2$, we see that this implies that $|\epsilon_n^{(k)}|\leq C n^{-\gamma_k}$ for some $\gamma_k>0$. We conclude that $\E(\Pi_n(q)^k)\to M_k(q)$.

\end{proof}

\section{Convergence of $\Pi_n(q)$ for $\frac{1}{2}<q<1$}

One can check that for $\frac{1}{2}<q<1$, $T(pq)-pT(q)>0$ for any $p>0$ so by the previous section, $\E(\Pi_n(q)^k)$ converges to $M_k(q)$ for all positive integers $k$. The natural question to ask then is does this imply that $\Pi_n(q)$ converges in law. This is the case if the moments $M_k(q)$ uniquely determine a probability measure. This in turn can be checked by Carleman's condition, or the following slightly weaker condition (for details, see e.g. Theorem 3.3.12 and the discussion around it in \cite{dur})

\begin{proposition}\label{prop:momprob}
Suppose the moments of the random variables $X_n$ satisfy $\E(X_n^k)\to \mu_k$ as $n\to\infty$ for each $k$ and
\begin{equation}
\limsup_{k\to\infty}\frac{1}{2k}\mu_{2k}^{\frac{1}{2k}}<\infty.
\end{equation}
Then $X_n$ converges in law to the unique distribution with moments $\mu_k$.
\end{proposition}

To check this property for the quantities $M_k(q)$, we prove the following proposition:

\begin{proposition}
There exists some $C(q)>0$ (independent of $k$) such that for $\frac{1}{2}<q<1$,

\begin{equation}
M_k(q)\leq (C(q))^k k!.
\end{equation}

\end{proposition}
\begin{proof}
Let us make the ansatz $M_j(q)\leq C^j j!$ for $j<k$. We then have by the recursion for $M_k$

\begin{align*}
M_k(q)&\leq \frac{C^k k!}{T(qk)-kT(q)}\frac{1}{2}\int_0^{\frac{\pi}{4}}\cos^{2qk-4}\theta\sum_{j=1}^{k-1}\tan^{2qj-2}\theta d\theta\\
&\leq C^k k!\frac{k}{T(qk)-kT(q)}\frac{1}{2}\int_0^{\frac{\pi}{4}}\cos^{2qk-4}\theta \tan^{2q-2}\theta d\theta.
\end{align*}

Recalling that

\begin{equation}
T(q)=\frac{\sqrt{\pi}}{2}\frac{\Gamma(q-\frac{1}{2})}{\Gamma(q-1)},
\end{equation}

\noindent we see (e.g. by Stirling's approximation) that $T(qk)\sim \sqrt{k}$ for large $k$ so that $k/(T(qk)-kT(q))\sim -1/T(q)>0$ for $q\in(1/2,1)$. Moreover, by say the monotone or dominated convergence theorem, the integral tends to zero as $k\to\infty$. Thus there exists a $k_0$ such that for $k\geq k_0$,

\begin{equation}
\frac{k}{T(qk)-kT(q)}\frac{1}{2}\int_0^{\frac{\pi}{4}}\cos^{2qk-4}\theta \tan^{2q-2}\theta d\theta<1.
\end{equation}

Also let us define

\begin{equation}
C(q)=\max_{j<k_0}\left(\frac{M_j(q)}{j!}\right)^{\frac{1}{j}}
\end{equation}

\noindent so that $M_j(q)\leq C(q)^j j!$ for $j<k_0$. Our argument and choice of $k_0$ then implies that $M_k(q)\leq C(q)^k k!$ for all $k$.
\end{proof}

Making use of Stirling's approximation, it is easy to check that this estimate is enough for us to be able to apply Proposition \ref{prop:momprob} and we get:

\begin{theorem}
For $\frac{1}{2}<q<1$, and any $b>0$, $\Pi_n(q)$ converges in law to the unique non-negative random variable whose positive integer moments are given by $(M_k(q))_k$.
\end{theorem}

\begin{remark} As the moments $M_k$ are independent of $b$ and the random variable whose moments these are is uniquely determined by them, the random variable is independent of $b$, so asymptotically, the only $b$ dependence in $P_n(q)$ is deterministic.
\end{remark}

\section{The stationary equation for $q\in\left(\frac{1}{2},1\right)$}

In addition to describing the positive integer moments, one can also hope to describe for example the Laplace transform of the limiting law. Let us thus define $\phi_n(t)=\E(e^{-t\Pi_n(q)})$ and let $\phi(t)$ be the Laplace transform of the law of the limit (i.e. $\phi(t)=\lim_{n\to\infty} \phi_n(t)$). The main goal of this section is to prove the following proposition

\begin{proposition}
For $\frac{1}{2}<q<1$, $\phi(t)$ is the unique solution to the equation

\begin{equation}\label{eq:stat}
T(q)t\phi'(t)+\frac{1}{2}\int_0^{\frac{\pi}{4}}\frac{d\theta}{\sin^2\theta\cos^2\theta}\left(\phi(t\sin^{2q}\theta )\phi(t\cos^{2q}\theta)-\phi(t)\right)=0
\end{equation}

\noindent in the space of Laplace transforms of laws of non-negative random variables with unit mean.

\end{proposition}

\begin{remark}
Note that our bound on the growth rate of the moments implies that $\phi$ is analytic in some domain around the origin and this integral equation can be obtained from the recursion for the moments in this region. Nonetheless, we shall provide a proof valid for all $t>0$.
\end{remark}

\begin{proof}
Uniqueness follows from the fact that this integral equation implies that the positive integer moments satisfy the recursion of the previous sections and we saw that these moments determine the random variable uniquely.

\vspace{0.3cm}

Let us write the recursion in terms of $\phi_n$ in the following way:

\begin{equation}
\phi_{n+1}(t)-\phi_n(t)=\E\left(\phi_n\left(t\frac{\sin^{2q}\theta_n}{T_n(q)}\right)\phi_n\left(t\frac{\cos^{2q}\theta_n}{T_n(q)}\right)-\phi_n(t)\right).
\end{equation}

Let us expand the integrand here:

\begin{align}
\notag \phi_n\left(t\frac{\sin^{2q}\theta}{T_n(q)}\right)&\phi_n\left(t\frac{\cos^{2q}\theta}{T_n(q)}\right)-\phi_n(t)
\\
\notag &=\phi_n(t\sin^{2q}\theta)\phi_n(t\cos^{2q}\theta)-\phi_n(t)\\
&+\phi_n\left(t\frac{\sin^{2q}\theta}{T_n(q)}\right) \phi_n'(t_{n,c}) t\cos^{2q}\theta\left(\frac{1}{T_n(q)}-1\right)\\
&\notag +\phi_n(t\cos^{2q}\theta)\phi_n'(t_{n,s})t \sin^{2q}\theta\left(\frac{1}{T_n(q)}-1\right)\notag,
\end{align}

\noindent for some $t_{n,c}$ between the points $t\cos^{2q}\theta$ and $t\frac{\cos^{2q}\theta}{T_n(q)}$ and $t_{n,s}$ between the points $t\sin^{2q}\theta$ and $t\frac{\sin^{2q}\theta}{T_n(q)}$.

\vspace{0.3cm}

We then note that the proof of Lemma \ref{sing} can be applied to functions depending on $n$ as well as long as we have uniform bounds on the function as well as the uniform estimate $f(\theta)=\mathcal{O}(|\theta|^\alpha)$ as $\theta\to 0$. Now it is easy to check that on $[-\pi/4,\pi/4]$ one has

\begin{equation}
|\phi_n(t\cos^{2q}\theta)\phi_n(t\sin^{2q}\theta)-\phi_n(t)|\leq t\sin^{2q}\theta+t(1-\cos^{2q}\theta)
\end{equation}

\noindent so an application of Lemma \ref{sing} gives

\begin{align}
\notag \E&(\phi_n(t\cos^{2q}\theta_n)\phi_n(t\sin^{2q}\theta_n)-\phi_n(t))\\
&=2\frac{b}{n}\int_0^{\frac{\pi}{4}}\frac{d\theta}{\sin^2 2\theta}\left(\phi_n(t\cos^{2q}\theta)\phi_n(t\sin^{2q}\theta)-\phi_n(t)\right) +\mathcal{O}\left(\left(\frac{b}{n}\right)^{\alpha}\right)
\end{align}

\noindent for a suitable $\alpha>1$. Note that by the uniform bounds on the integrand, the integral converges to the corresponding one with $\phi_n$ replaced by $\phi$.

\vspace{0.3cm}

For the other two terms, consider the function

\begin{align}
\notag f(\theta)&=-tT(q)\phi_n'(t)+\phi_n\left(t\frac{\sin^{2q}\theta}{T_n(q)}\right) \phi_n'(t_{n,c}) t\cos^{2q}\theta\\
&\notag +\phi_n(t\cos^{2q}\theta)\phi_n'(t_{n,s})t \sin^{2q}\theta.
\end{align}

Noting that as $\phi_n''$ is uniformly bounded (as $q<1$), we see that as $\theta\to 0$,

\begin{equation}
f(\theta)=\mathcal{O}(|\theta|^{\min(2,2q)})+\mathcal{O}\left(\frac{b}{n}\right),
\end{equation}

\noindent where both bounds are uniform in $n$. Thus making use of Lemma \ref{sing} once again,

\begin{align}
\notag \E&\Bigg(\phi_n\left(t\frac{\sin^{2q}\theta}{T_n(q)}\right) \phi_n'(t_{n,c}) t\cos^{2q}\theta+\phi_n(t\cos^{2q}\theta)\phi_n'(t_{n,s})t \sin^{2q}\theta\Bigg)\\
&=tT(q)\phi_n'(t)+\mathcal{O}\left(\frac{b}{n}\right),
\end{align}

\noindent for some $\alpha>1$. We conclude that

\begin{align}
\notag &\frac{n}{b}(\phi_{n+1}(t)-\phi_n(t))\\
&=\frac{1}{2}\int_0^{\frac{\pi}{4}}\frac{d\theta}{\sin^2 2\theta}(\phi_n(t\cos^{2q}\theta)\phi_n(t\sin^{2q}\theta)-\phi_n(t))+tT(q)\phi_n'(t) \\
& \notag+\mathcal{O}\left(\left(\frac{b}{n}\right)^{\alpha-1}\right),
\end{align}

\noindent for some $\alpha>1$. Moreover, the bound is uniform in $n$.

\vspace{0.3cm}

Now the left side converges as $n\to\infty$ (the integral by the dominated convergence theorem and the bounds we pointed out, and the derivative by the dominated convergence theorem when we write it as $-\E(\Pi_n e^{-t\Pi_n})$). The left side can only converge to zero as otherwise $\phi_n$ would diverge logarithmically, which is impossible. Thus we see that $\phi$ satisfies the equation in the claim.

\end{proof}

\section{Discussion and open problems}

In this section we point out some issues we were not able to address and discuss some possible further questions related to the problem studied in this paper.

\subsection{Connection to band matrices}
In Section \ref{sec:model} (see Remark \ref{imbrie}) we discussed a possible approach to going beyond the LME approximation for the band matrix model. A natural question would be to find an actual random matrix model for which a LME-type recursion becomes exact. For the analogous problem of Gaussian multiplicative chaos measures, i.e. measures whose density is given by exponential of the Gaussian Free Field the role of the LEM recursion is played by the Mandelbrot cascades resulting from replacing the GFF by a hierarchical version thereof. Note that the critical random matrix models with hierarchical structure do exist \cite{FOR2009,MG2011}, but that structure itself is not enough
to make the LME recursion exact. Is there a hierarchical random matrix that produces a LME recursion?
\vspace{0.3cm}

\subsection{Convergence for $1<q<q_c$}
Based on the analogy with the branching random walk
and the bounds on the moments $\E(\Pi_n(q)^p)$ for some $p>1$, it is natural to expect that $\Pi_n(q)$ converges also in this regime. For $q<1$, our approach centered around the fact that all moments existed and we were able to describe their limiting values through a recursion.

For $q\geq 1$, all moments will no longer exist, but the stationary equation \eqref{eq:stat} still makes sense. For branching Brownian motion Bramson \cite{bram} (see also \cite{cw,ds} for the case of branching random walks) proved convergence by perturbing around stationary solution to the stationary equation. However, in that case monotonicity due to maximum principle was an
important ingredient and we have not been able to prove such monotonicity in the LME case.

\subsection{Properties of the limiting random variable} In the case of the branching random walk a lot is known about
the limiting random variable:
it is almost surely positive, it has (in the Gaussian case) all negative moments and the tail of its distribution is known. In our case, our moment estimates suggest that the tail of the distribution might be similar i.e. we conjecture that for $\frac{1}{2}<q<q_c$ the limit
\begin{equation}
\lim_{x\to\infty}x^{p^*(q)}\mathbb{P}(\Pi(q)\geq x)
\end{equation}
exists.
Moreover, the stationary equation suggests that if one writes $\Pi(q)$ for the limiting random variable, one should have $\mathbb{P}(\Pi(q)=0)=\phi(\infty)=0$ and in analogy with BRW it is natural to conjecture  $\Pi(q)$ almost surely positive and has all  negative moments.

\subsection{$q\geq q_c$ and the freezing transition} The analogy with  branching random walk leads one to ask whether $\Pi_n(q_c)$ converges to zero and if so whether a nontrivial limit can be obtained with additional normalization.  In the study of the branching random walk, $\Pi_n(q_c)$ converges to zero under the case of a "second order phase transition" or the so-called "boundary case" (see e.g. \cite{dl,big}) which in our model should correspond to $q_cT'(q_c)-T(q_c)=0$. There are branching random walks where this is not satisfied (though the critical point is characterized slightly differently) and their behavior for $q\geq q_c$ is significantly different from the "Gaussian case" see e.g. \cite{bhm,ab}. In this case, $\Pi_n(q_c)$ converges to a non-trivial random variable without any further normalization.  Thus for $q_c$, the central question is does $\Pi_n(q_c)$ converge to a non-trivial random variable or is the correct quantity to study $-\partial_q \Pi_n(q_c)$ (which can be defined e.g. using the tree construction).

For the branching random walk  the case corresponding to $q>q_c$ is controlled by the behavior of the maximum of the walk. In our case, using the tree construction, we can  define a branching random walk which however  is inhomogeneous and also the branches repulse each other. This leads to peculiar behavior. Indeed,  consider the law of $(\log \sin^2\theta_n, \log \cos^2\theta_n)$. With probability of order $1-\frac{C}{n}$,  $\theta_n$ is of order $\frac{1}{n}$ so $\log \sin^2\theta$ is of order $-\log n$ while $\log \cos^2\theta$ is of order $-\frac{1}{n^2}$. With probability of order $\frac{1}{n}$, $\theta_n$ is of order one so that both $\log\sin^2\theta_n$ and $\log\cos^2\theta_n$ are also of order one. This picture where apart from a rare event, the path of the maximum is deterministic is rather different from the standard branching random walk picture.
While studying this maximum is a non-trivial question, we refer to \cite{fg}, where it is argued that  the scaling of the maximum is essentially the same as for a branching random walk with independent Gaussian increments.

Finally, it is natural to ask is there a freezing transition phenomenon in the model. If there is a normalization under which $\Pi_n(q)$ converges to say $\widetilde{\Pi}(q)$ for $q>q_c$, is this given in terms of the relevant critical random variable and a stable L\'evy subordinator?

\subsection{Geometry of the eigenvectors}
The eigenvectors of the random matrix are given by
$$
\psi_i^N=O_n\psi_i^0,\ \ i=0,\dots N-1
$$
where $(\psi_i^0)_j=\delta_{ij}$.
As in the derivation of the recursion \eqref{lemrec} the recursion   \eqref{evrec} on the resonant set can be used to arrive to the following recursion for  the eigenvector. Let $\theta_{n,i}^\pm$ be the angles determined by the matrices $M_{ij }^n$ of \eqref{reson} with $j=i\pm(n+1)$. Fixing $i$ either $\{i,i+n+1\}$ or  $\{i,i-n-1\}$  belongs to the resonant set $R_n$, with equal probability. Let $\sigma_i^n\in\{0,1\}$ be a Bernouilli $(\hf,\hf)$ random variable. The model for the eigenvector recursion is
\baq\nonumber
\psi_i^{n+1}&=&(\sigma_i^n\cos \theta_{n,i}^+ +(1-\sigma_i^n)\cos \theta_{n,i}^-)\psi_i^{n}\\
&&+\sigma_i^n\sin \theta_{n,i}^+\psi_{i+n+1}^{n}+(1-\sigma_i^n)\sin \theta_{n,i}^-\psi_{i-n-1}^{n}\nonumber
\eaq
 This leads to the representation
 \beq\label{rw}
\psi_i^N=\sum_\omega\prod_{n=0}^{N-1}p_n(\omega_n,\omega_{n+1})
\eeq
where the sum is over all paths $\omega=(\omega_0,\omega_1,\dots,\omega_{N})$ on $\Z_N$ with  $\omega_0=i$, $\omega_{N}=j$ and
 $\omega_{n+1}\in\{\omega_n, \omega_n+n+1, \omega_n-n-1\}$%\pm n\ \ {\rm mod}\ \ N$
 and
$$
p_n(\omega_n,\omega_{n+1})=\left\{\begin{array}{ll}
\sigma_{\omega_n}^n\cos \theta_{n,{\omega_n}}^+ +(1-\sigma_{\omega_n}^n)\cos \theta_{n,{\omega_n}}^-
  & \omega_{n+1}=\omega_n\\
  \sigma_{\omega_n}^n\sin \theta_{n,{\omega_n}}^+& \omega_{n+1}=\omega_n
  + n+1%\ \ {\rm mod}\ \ N
   \\
  (1-\sigma_{\omega_n}^n)\sin \theta_{n,{\omega_n}}^-
 &  \omega_{n+1}=\omega_n- n-1%\ \ {\rm mod}\ \ N
 \end{array} \right.
$$
and $\theta_{n,i}$, $i\in\Z_N$ are i.i.d. copies of $\theta_n$.\\
 It remains a challenge to use such a description for a more detailed understanding of the structure of eigenvectors.


\begin{thebibliography}{99}
\bibitem{aid}  E. A\"id\'ekon: Convergence in law of the minimum of a branching random walk. Ann. Probab. {\bf 41} (2013), no. 3A, 1362–1426.
\bibitem{as} E. A\"id\'ekon and Z. Shi: The Seneta-Heyde scaling for the branching random walk. Ann. Probab. {\bf 42} (2014), no. 3, 959–993.
\bibitem{ABH} L.-P. Arguin, D. Belius, A.J. Harper.  Maxima of a randomized Riemann zeta function,
and branching random walks. Arxiv:1506.00629
\bibitem{ab} N. Attia, J. Barral: Hausdorff and packing spectra, large deviations, and free energy for branching random walks in $R^d$. Comm. Math. Phys. {\bf 331} (2014), no. 1, 139–187.
\bibitem{bkm} Bacry E., Kozhemyak, A., Muzy J.-F.: Continuous cascade models for asset returns,
J. Econom. Dynam. Control {\bf 32} (2008), no. 1, 156–199.
\bibitem{bm} J. Barral and B. Mandelbrot: Fractional multiplicative processes. Ann. Inst. H. Poincar\'e. {\bf 45} (2009) no. 4, 1116-1129.
\bibitem{bhm}J. Barral, Y. Hu, and T. Madaule: The minimum of a branching random walk outside the boundary case. arXiv:1406.6971
    \bibitem{brv} J. Barral, R. Rhodes, and V. Vargas: Limiting laws of supercritical branching random walks. C. R. Math. Acad. Sci. Paris {\bf 350} (2012), no. 9-10, 535–538.
\bibitem{big} J. D. Biggins: Martingale convergence in the branching random walk. Jnl. Appl. Probab {\bf 14} (1977), 25-37.
\bibitem{bk} J.D. Biggins and A.E. Kyprianou: Measure change in multitype branching. Adv. Appl. Probab. {\bf 36}, 544-581, 2004.
\bibitem{BG2011} E. Bogomolny, O. Giraud: Eigenfunction entropy and spectral
compressibility for critical random matrix ensembles. Phys. Rev.
Lett. {\bf 106} (2011) 044101.
\bibitem{bram} M. Bramson: Convergence of solutions of the Kolmogorov equation to travelling waves. Mem. Amer. Math. Soc. {\bf 44} (1983), no. 285, iv+190 pp.
\bibitem{tailc} D. Buraczewski: On tails of fixed points of the smoothing transform in the boundary case. Stochastic Process. Appl. {\bf 119} (2009), no. 11, 3955–3961
\bibitem{cld} D. Carpentier and P. Le Doussal: Glass transition of a particle in a random potential, front selection in non linear RG and entropic phenomena in Liouville and SinhGordon models Phys. Rev. E {\bf 63}, 026110 (2001).
\bibitem{ds} B. Derrida and H. Spohn: Polymers on disordered trees, spin glasses and travelling waves. J. Stat. Phys. {\bf 51}, 817-840, 1988.
\bibitem{DRZ2015} J.~Ding, R.~Roy and O.~Zeitouni: Convergence of the centered maximum of log-correlated Gaussian fields. arXiv:1503.04588 (2015)
\bibitem{dur} R. Durrett: Probability: theory and examples. Fourth edition. Cambridge Series in Statistical and Probabilistic Mathematics. Cambridge University Press, Cambridge, 2010. x+428 pp. ISBN: 978-0-521-76539-8.
\bibitem{dl} R. Durrett and T.M. Liggett:  Fixed points of the smoothing transformation. Z. Wahrsch. Verw. Gebiete {\bf 64} (1983), no. 3, 275–301.
\bibitem{EMrev} F. Evers, A.D. Mirlin: Anderson transitions. Rev Mod Phys {\bf 80} 2008,
1355–417.
\bibitem{FyoRev2010} Y.V. Fyodorov: Multifractality and freezing phenomena in random energy landscapes:
An introduction. Physical A {\bf 389 } (2010)  4229-4254.
\bibitem{FB08} Y.~V. Fyodorov and J.~P. Bouchaud: Freezing and extreme-value statistics in a random energy model with
  logarithmically correlated potential. J. Phys. A: Math. Theor. {\bf 41} (2008), no. 37, 372001, 12pp
\bibitem{FLDR2009} Y.~V. Fyodorov, P.~Le~Doussal, and A.~Rosso: Statistical Mechanics of Logarithmic REM: Duality, Freezing and Extreme Value Statistics of 1/f Noises Generated by Gaussian Free Fields.
J. Stat. Mech. Theory Exp. {\bf 2009} (2009), no. 10, P10005, 32 pp
\bibitem{FLDR2012} Y.~V. Fyodorov, P.~Le~Doussal, and A.~Rosso: Counting Function Fluctuations and Extreme Value Threshold in Multifractal Patterns: The Case Study of an Ideal 1/f Noise.
J. Stat. Phys. {\bf 149} (2012), 898--920
\bibitem{fg} Y. V. Fyodorov and O. Giraud: High values of disorder-generated multifractals and logarithmically correlated processes.  Chaos, Solitons $\&$ Fractals {\bf 74} (2015) 15–26
\bibitem{FyoKeat2014} Y.~V. Fyodorov and J.P. Keating: Freezing Transitions and Extreme Values: Random Matrix
  Theory, $\zeta(1/2+it)$ and Disordered Landscapes.
Philos. Trans. R. Soc. Lond. Ser. A Math. Phys. Eng. Sci. {\bf 372} (2014) no. 2007, 20120503, 32pp
\bibitem{FOR2009} Y.V. Fyodorov, A. Ossipov and A. Rodriguez: The Anderson localization transition
and eigenfunction multifractality in an ensemble of ultrametric random
matrices. J. Stat. Mech. {\bf 2009} (2009) L12001
\bibitem{tailsubc} Guivarc'h, Y: Sur une extension de la notion de loi semi-stable. Ann. Inst. H. Poincaré Probab. Statist. {\bf 26} (1990), no. 2, 261–285.
\bibitem{Imbrie} J. Z. Imbrie, Multi-Scale Jacobi Method for Anderson Localization. Comm. Math. Phys. 341 (2016), no. 2, 491–521.
\bibitem{kah} J.-P. Kahane: Sur le modele de turbulence de Benoit Mandelbrot., C.R. Acad.
Sci. Paris, {\bf  278} (1974), 567–569.
\bibitem{kp} J.-P. Kahane and J. Peyri\`ere: Sur certaines martingales de B. Mandelbrot. Adv. Math. {\bf 22}, 131-145, 1976.
    \bibitem{Kravtsov} V.E. Kravtsov. Random matrix representations of critical statistics. arXiv:0911.0615
\bibitem{lev1} L.S. Levitov, Delocalization of vibrational modes caused by electric dipole interaction. Phys. Rev. Lett. {\bf 64}, 547, 1990.
\bibitem{lev2} L.S. Levitov, Critical Hamiltonians with long range hopping, Annalen der Physik, {\bf 8}, Issue 7, pp.697-706, 1999.
\bibitem{rl} R. Lyons:  A simple path to Biggins' martingale convergence for branching random walk. Classical and modern branching processes (Minneapolis, MN, 1994), 217–221, IMA Vol. Math. Appl., {\bf 84}, Springer, New York, 1997.
\bibitem{mad} T. Madaule: Convergence in law for the branching random walk seen from its tip. arXiv:1107.2543, 2011.
\bibitem{mandelbrot}  B.B. Mandelbrot: Intermittent turbulence in self-similar cascades, divergence of high
moments and dimension of the carrier, J. Fluid. Mech. {\bf 62} (1974), 331-358.
\bibitem{em} A.D. Mirlin and F. Evers: Multifractality and critical fluctuations at the Anderson transition, Phys. Rev. B {\bf 62}, 7920, 2000.
\bibitem{prbm} A.D. Mirlin, Y.V. Fyodorov, F.-M. Dittes, J. Quezada, and T.H. Seligman: Transition from localized to extended eigenstates in the ensemble of power-law random banded matrices, Phys. Rev. E {\bf 54}, 3221, 1996.
\bibitem{molch} G.M. Molchan: Scaling exponents and multifractal dimensions for independent random cascades, Comm. Math Physics {\bf 179} (1996), 681-702.
\bibitem{MG2011} C. Monthus and T. Garel. A critical Dyson hierarchical model for the Anderson localization transition J. Stat. Mech. (2011) P05005.  doi:10.1088/1742-5468/2011/05/P05005
\bibitem{ps} I.V.Ponomarev and P.G.Silvestrov: Coherent propagation of interacting particles in a random potential: the Mechanism of enhancement Phys.Rev. B v. 56, 3742-3759 (1997).
\bibitem{gmc} R. Rhodes and V. Vargas: Gaussian multiplicative chaos and applications: a review. Probab. Surv. 11 (2014) 315-392.
\bibitem{rderev} U. R\"osler and L. R\"uschendorf: The contraction method for recursive algorithms. Algorithmica {\bf 29} 3–33, 2001.
\bibitem{cw} C. Webb: Exact Asymptotics of the Freezing Transition of a Logarithmically Correlated Random Energy Model, J. Stat. Phys. {\bf 145}, 1595-1619, 2011.
\end{thebibliography}
\end{document}